\documentclass{svproc}
\usepackage{bm}
\usepackage{amsmath}
\usepackage{amsfonts}
\usepackage{graphicx}
\usepackage[colorinlistoftodos]{todonotes}
\usepackage{yhmath}  
\usepackage{mathtools}
\usepackage{tikz}
\usepackage{afterpage}
\usetikzlibrary{shapes,arrows}
\usepackage{indentfirst}
\usepackage[hidelinks]{hyperref}  
\usepackage[sectionbib,numbers]{natbib}

\begin{document}

\title{Analytic characterization of stability islands on two point vortex systems}

\author{
Gil M. Marques$^{1}$, Sílvio Gama$^{1}$ and Fernando L. Pereira$^{2}$}

\maketitle
\thispagestyle{empty}
\pagestyle{empty}
\setcounter{footnote}{1}
\footnotetext{Centro de Matemática da Universidade do Porto, Departamento de Matemática, Faculdade de Ciências, Universidade do Porto, Rua do Campo Alegre s/n, 4169-007, Porto, Portugal\\
        }
\setcounter{footnote}{2}
\footnotetext{{SYSTEC - Departamento de Engenharia Eletrotécnica e de Computadores, Faculdade de Engenharia da Universidade do Porto, Rua Dr. Roberto Frias, s/n, 4200-465, Porto, Portugal}}





\begin{abstract}
In a system of point vortices, there exist regions of stability around each vortex, even if the system is chaotic. These regions are usually called stability islands and they have a morphology that is hard to characterise. We study and characterise them in two point vortex systems in the infinite two-dimensional plane - the simplest scenario - by studying the dynamics of passive particles in these environments. We present computations for the perimeter and area of these islands and highlight the analytical expressions that define their boundary.

\end{abstract}
\section{Introduction}


Point vortices are singular solutions of the two-dimensional (2D) incompressible Euler equations which were first studied by Helmholtz~\cite{Helmholtz}, having, some years later, received the attention of Lord Kelvin~\cite{Kelvin1869} and Kirchhoff~\cite{Kirchhoff1876}. Since then, and until the present day, point vortices have become the subject of research in domains as varied as torus~\cite{Makoto2006,Stremler2010}, spheres~\cite{crowdy2006point,newton2010n,mokhov2020point}, hyperbolic surfaces~\cite{Nava2014,Ragazzo2017}, etc.. They correspond to a scenario where the vorticity of a flow is concentrated in some well-defined points in space and enable a simpler description of the dynamics of such a system. 
Knowing the position and strength (or circulation) of each vortex in the system suffices to obtain the full velocity field at that instant and thus, knowing the dynamics of the vortices themselves in some time interval is enough to characterize the velocity field during that same time interval. The dynamics of these vortices follow ordinary differential equations that are akin to the ones in the gravitational $n$-body problem and their trajectories can be found analytically for some specific spacial configurations, such as relative equilibria \cite{Aref_Playground}. Consider a 2D inviscid flow in the complex plane that is described by $n$ point vortices located in the positions $z_{\alpha}=x_{\alpha} + iy_{\alpha}\,(\alpha=1,\dots,n)$. The motion of vortices is defined by the inial value problem determined by 

\begin{equation}
	\dot{z}_{\alpha}^*=\frac{1}{2\pi i}\sum_{\substack{\beta=1\\ \beta\neq \alpha}}^{n}\frac{\Gamma_{\beta}}{z_{\alpha}-z_{\beta}}, \quad \alpha=1,\dots,n,
 \label{eq:motion}
\end{equation}
and the initial positions of the vortices. Here, $\Gamma_{\alpha}$ is the circulation of the $\alpha^{th}$ vortex. Moreover, since the flow is incompressible, we know that the information of the velocity field can also be represented in a stream function formulation with 
\begin{equation}
	\psi\left(z,t\right) = -\frac{1}{4\pi}\sum_{\alpha=1}^{n}\Gamma_{\alpha}\log|z-z_{\alpha}\left(t\right)|.
\end{equation}
To retrieve the velocity field it suffices to consider the partial derivatives of this stream function: $\dot{x} = \dfrac{\partial\psi}{\partial y}$, $\dot{y} =-\dfrac{\partial\psi}{\partial x}$, or, in complex notation, $\dot{z}^* = 2i\,\dfrac{\partial\psi}{\partial z}$.

Studying the velocity field that arises from a system of point vortices can be thought of as studying the trajectories of passive particles that move in the velocity field. A passive particle is a particle that does not change the flow itself but is transported by it; it can be thought of as a point vortex with zero strength and, thus, its' motion follows the same differential equations as the vortices. This is analogous to the restricted $n+1$ body problem in celestial mechanics where one of the bodies of the system is assumed to have zero mass. More explicitly, the equation of motion for a passive particle is

\begin{equation}
	\dot{z}^*=\frac{1}{2\pi i}\sum_{\alpha=1}^{n}\frac{\Gamma_{\alpha}}{z-z_{\alpha}}.
 \label{eq:motion_particle}
\end{equation}

It is known that around each vortex there exists what it is usually called an island of stability, where the motion of passive particles is always regular and is mostly ruled by the associated vortex, while suffering little influence from the other $n-1$ vortices \cite{Babiano}. Outside of these islands, one expects the movement of particles to be chaotic for $\geq4$, except on very specific configurations. The boundaries of these islands constitute barriers to particle motion and, thus, to fluid transportation: particles inside of these regularity islands cannot be ejected from them and particles moving in the background between the vortices cannot enter the stability regions associated with each vortex.

In this work we will cover the simplest case $n=2$. Even though in this scenario every particle trajectory is regular, the presence of the vortices is enough to create rigid boundaries in the 2D plane from where particles are not able to escape. Furthermore, particles inside one region present similar trajectories and are generally different from particle trajectories in another one of these regions. The regions in which the vortices are contained are a simpler version of the stability islands that arise in the chaotic $n\geq4$ scenario and we present computations of their perimeter and area, as well as analytical expressions that define their boundary.

\section{Dynamics of Passive Particles on Two Point Vortex Systems}\label{sec:dynamics}

We now consider a system comprised by only two vortices located, at time $t$, in $z_1\left(t\right)$, $z_2\left(t\right)$ with circulations $\Gamma_1$ and $\Gamma_2$ respectively. There are two different regimes for the motion of the vortices depending on if the sum of their circulations is zero or non-zero. In the following two subsections, we will treat these two cases separately, describing the dynamics of passive particles in the plane and analyzing the boundaries that separate the plane in different regions from which particles cannot escape.

\subsection{Case $\quad \Gamma_1 +\Gamma_2 = 0$}

We consider two vortices with circulations $\Gamma_1 = - \Gamma_2 = k \in\mathbb{R}^+$ (if $k \in\mathbb{R}^-$, the motion of the system will simply happen in the opposite direction). It is known that in such a system, both vortices will move in straight lines, parallel to one another with a constant velocity \cite{Batchelor}. Without any loss of generality, it is possible to consider a system where the vortices are initially on the imaginary axis and separated by a distance $d\in\mathbb{R}^+$, i.e., $z_1\left(0\right) = \dfrac{d}{2}\,i,\,\, z_2\left(0\right) = -\dfrac{d}{2}\,i$.

The solution of equation (\ref{eq:motion}) for this case are thus 

\begin{equation}
\left\{
\begin{aligned}
	&z_1\left(t\right) = \frac{d}{2}\,i + \frac{k}{2\pi d}\,t\\
	&z_2\left(t\right) = -\frac{d}{2}\,i + \frac{k}{2\pi d}\,t
\end{aligned}
\right..
\end{equation}

The equation of motion for a passive particle in this system is non-autonomous and given by

\begin{equation}
\begin{aligned}
	\dot{z}^*&=\frac{k}{2\pi i}\left[\frac{1}{z-z_1\left(t\right)} - \frac{1}{z-z_2\left(t\right)}\right]\\
	&=\frac{k}{2\pi i}\left[\frac{1}{z-\frac{k}{2\pi d}t - \frac{d}{2}i} -\frac{1}{z-\frac{k}{2\pi d}t + 	\frac{d}{2}i} \right].
\end{aligned}
\end{equation}

Now, consider the change of coordinates $w = \dfrac{z}{d} - \dfrac{k}{2\pi d^2}\,t$. In this frame, the vortices are stationary and their positions are $w_1\left(t\right)=i/2$, $w_2\left(t\right)=-i/2$. The equation of motion of a passive particle is now governed by the autonomous differential equation

\begin{equation}
	\dot{w}^* = \frac{k}{2\pi d^2i}\left[\frac{1}{w-i/2} - \frac{1}{w+i/2} - i\right].
\label{eq:dw_a}
\end{equation}

Our aim is to describe the trajectories of passive particles. In this co-moving frame, stagnation points are the simplest possible trajectories and correspond to particles that move with the same velocity as the vortices in the original frame. These correspond to zeros of $\dot{w}$ and it is easy to see that the system has only two of them: $w=\pm\sqrt{3}/2$. These correspond to the Lagrangian points of the system and are real solutions of Eq. (\ref{eq:dw_a}), but they are not the only real solutions for this equation. Other real solutions $x\left(t\right)$ of Eq. (\ref{eq:dw_a}) must satisfy the equation


\begin{equation}
	\dot{x} = \frac{k}{2\pi d^2}\left[\frac{4}{4x^2+1}-1\right].
\label{eq:real_diff}
\end{equation}

It is possible to integrate this equation and find an implicit expression for the non-stationary solutions of Eq. (\ref{eq:real_diff}):

\begin{multline}
	\frac{1}{\sqrt3}\left[\log\left(\frac{2x\left(t\right) + \sqrt3}{2x\left(0\right) + \sqrt3}\right) - \log\left(\frac{2x\left(t\right) - \sqrt3}{2x\left(0\right) - \sqrt3}\right) \right] = x\left(t\right) - x\left(0\right) + \frac{k}{2\pi d^2}t, \\ x\left(0\right)\in\mathbb{R} \setminus\left\{\sqrt{3}/2,-\sqrt{3}/2\right\}.
\end{multline}

Notice that this equation defines three different solutions of Eq. (\ref{eq:real_diff}), depending on the value of $x\left(0\right)$ and that these three solutions and the stagnation points make up the full real axis. This means that the real axis is a physical barrier to the motion of passive particles, as, by the existence and uniqueness of solutions theorem for ordinary differential equations, there can not exist a passive particle trajectory that crosses it. Furthermore, we can see that as $t\to+\infty$:
\begin{itemize}
    \item if $x\left(0\right)<-\sqrt{3}/2$, then $x\left(t\right)\to-\infty$, 
    \item if $x\left(0\right)>\sqrt{3}/2$, then $x\left(t\right)\to\frac{\sqrt{3}}{2}^+$, and
    \item if $|x\left(0\right)|<\sqrt{3}/2$, then $x\left(t\right)\to\frac{\sqrt{3}}{2}^-$.
\end{itemize}
Thus, solutions in the real axis are attracted to $\sqrt{3}/2$ and repelled from $-\sqrt{3}/2$.

To understand the motion of passive particles in the rest of the 2D plane, we rewrite Eq. (\ref{eq:dw_a}) as a system of two real ODEs $\left(\dot{x},\dot{y}\right) = F\left(x,y\right)\,,$

\begin{equation}
\left\{
\begin{aligned}
	\dot{x} &= -&\frac{y-\frac{1}{2}}{x^2 + \left(y-\frac{1}{2}\right)^2} + &\frac{y+\frac{1}{2}}{x^2 + \left(y+\frac{1}{2}\right)^2} - 1\\
	\dot{y} &= &\frac{x}{x^2 + \left(y-\frac{1}{2}\right)^2} - &\frac{x}{x^2 + \left(y+\frac{1}{2}\right)^2}
\end{aligned}
\right.,
\label{eq:gamma0_2d}
\end{equation}
where we have rescaled time $t\to\dfrac{k}{2\pi d^2}\,t$ for simplicity.

We first analyze the stability of the Lagrangian points by checking the eigenvalues and eigenvectors of the linearization of $F\left(x,y\right)$ on those points. Defining the quantities $\ell_+\left(x,y\right) = x^2+\left(y + \frac{1}{2}\right)^2$, $\ell_-\left(x,y\right) = x^2+\left(y - \frac{1}{2}\right)^2$, we have

\begin{equation}
	DF\left(x,y\right) = \begin{pmatrix}
		-\frac{2x\left(y+\frac{1}{2}\right)}{\ell_+^2} + \frac{2x\left(y-\frac{1}{2}\right)}{\ell_-^2}  & \quad\quad\frac{1}{\ell_+} - \frac{1}{\ell_-} - \frac{2\left(y+\frac{1}{2}\right)^2}{\ell_+^2} + \frac{2\left(y-\frac{1}{2}\right)^2}{\ell_-^2}\\
		-\frac{1}{\ell_+} + \frac{1}{\ell_-} - \frac{2x^2}{\ell_+^2} - \frac{2x^2}{\ell_-^2} & \quad\quad\frac{2x\left(y+\frac{1}{2}\right)}{\ell_+^2} - \frac{2x\left(y-\frac{1}{2}\right)}{\ell_-^2}
\end{pmatrix},
\end{equation}

\begin{equation}
	DF\left(\pm\frac{\sqrt3}{2},0\right) = 
	\begin{pmatrix}
		\mp\sqrt3	& 0 \\
		0 & \pm\sqrt3
	\end{pmatrix}.
\end{equation}

Thus, the stagnation points are both saddle points. Locally, $\left(-\frac{\sqrt3}{2},0\right)$ attracts solutions on the $\left(0,1\right)$ direction and repels them on the $\left(1,0\right)$ direction, while we observe the exact opposite behavior in a neighborhood of $\left(\frac{\sqrt3}{2},0\right)$.

Using a stream function formulation, the system (\ref{eq:gamma0_2d}) can be further rewritten as $\dot{x} = \dfrac{\partial\psi}{\partial y}$, $\dot{y} =-\dfrac{\partial\psi}{\partial x}$. The stream function for this system is time-independent and can be written as

\begin{equation}
	\psi\left(x,y\right) = -\frac{1}{2}\left[\log\left(x^2+\left(y - \frac{1}{2}\right)^2\right) - \log\left(x^2+\left(y + \frac{1}{2}\right)^2\right) + 2y\right].
\end{equation}

Thus, the trajectories of passive particles can be identified as the level sets of $\psi\left(x,y\right)$ (this is only true if the stream function is time-independent~\cite{Pope}). Equivalently, this can be rewritten in a clearer manner as a modified stream function

\begin{equation}
	\widetilde{\psi}\left(x,y\right) = e^{-2\psi\left(x,y\right)} = \left[x^2+\left(y - \frac{1}{2}\right)^2\right] \left[x^2+\left(y + \frac{1}{2}\right)^2\right]^{-1} e^{2y}.
\end{equation}

Notice that for $\forall y : |y|>1/2,\,\,\,\dot{x}<1$. This, together with the fact that $\widetilde{\psi}\left(-\frac{\sqrt3}{2},0\right) = \widetilde{\psi}\left(\frac{\sqrt3}{2},0\right)$ means that there should be a trajectory of a passive particle linking the stagnation points other than the one on the real axis. These trajectories thus separate the infinite plane in 4 regions. A particle that starts its' trajectory in one region cannot cross into any of the other 3 regions. We can further assess that two of these regions are closed and each of them contains one of the vortices. In fact, it is possible to easily draw the phase diagram of the system with this information. The phase diagram is drawn on Figure \ref{fig:1}. 

\begin{figure}[h]
	\centering
	\includegraphics[width=\linewidth]{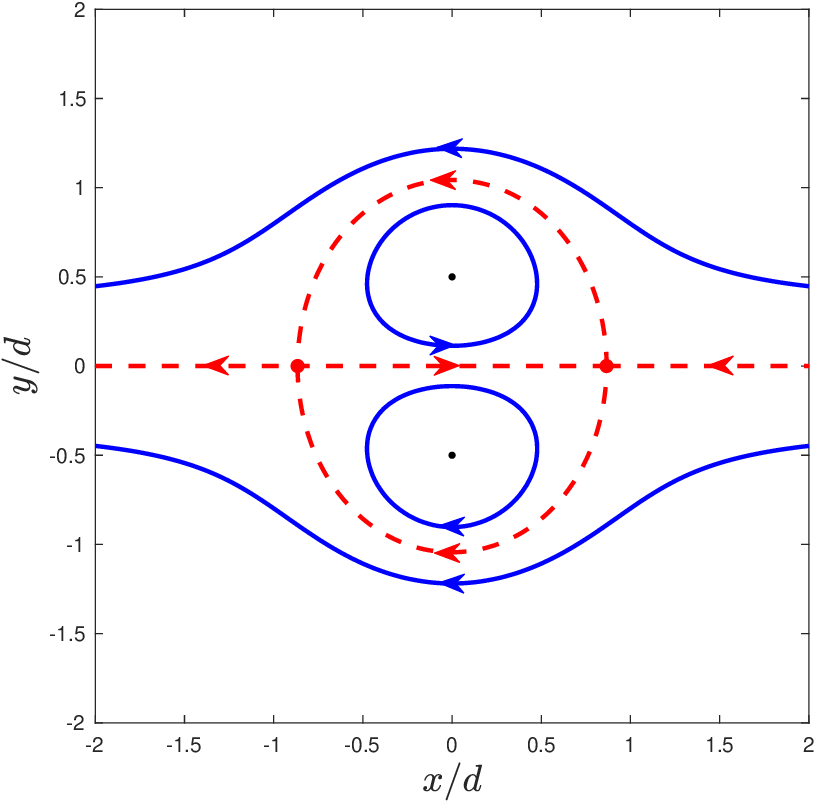}
	\caption[Phase Diagram.]{Phase diagram for a system of two point vortices in the co-moving frame. The vortices are located in $w_1=i/2$ and $w_1=-i/2$ and the circulations are such that $\Gamma_1 +\Gamma_2 = 0$. The red dots identify the stagnation points of the system. The red dashed lines correspond to special particle trajectories that constitute a boundary between regions of the infinite plane where particles present different types of trajectories (exemplified in blue lines).}
	\label{fig:1}
\end{figure}	

The boundary between the different regions can be defined implicitly using the modified stream function as $\widetilde{\psi}\left(x,y\right) = \widetilde{\psi}\left(0,0\right) = 1$, or, more explicitly:

\begin{equation}
	\left[x^2+\left(y - \frac{1}{2}\right)^2\right] e^{2y} = x^2+\left(y + \frac{1}{2}\right)^2.
\label{eq:frontier_a}
\end{equation}

The two closed regions defined by Eq. (\ref{eq:frontier_a}) are called the stability islands of the vortices and their existence has been shown in simulations numerous works \cite{Babiano,boatto_pierrehumbert_1999} for systems of any number of vortices.

It is also possible to obtain an expression for the area of each of these closed regions. Notice that the $x$- and $y$-axis are also symmetry axes for the solutions of the system. Thus, we only need to calculate the area enclosed in the region of interest that lies on the first quadrant, i. e., the area below the curve

\begin{equation}
	\xi= \left\{\left( x,y\right)\in \mathbb{R}^+\times\mathbb{R}^+ : \left[x^2+\left(y - \frac{1}{2}\right)^2\right] e^{2y} = \left[x^2+\left(y + \frac{1}{2}\right)^2\right]\right\}.
\end{equation}

It is possible to rewrite Eq. (\ref{eq:frontier_a}) as $x^2\left(y\right) = -y^2+y\,\coth\left(y\right)-\frac{1}{4}$ and by finding (numerically) some $\widetilde{y}>0$ such that $x^2\left(\widetilde{y}\right)=0$, the area of each stability island can be calculated as

\begin{equation}
	A = 2\int_{0}^{\widetilde{y}} \sqrt{-y^2+y\,\coth\left(y\right)-\frac{1}{4}}\,dy = 1.42505829613753 \pm 5\times10^{-14},
\end{equation}
and thus, for a system where the vortices are separated by a distance $d$, the area of each stability island is $A\,d^2\approx1.425\,d^2$.

It is also possible to compute the perimeter of each stability island. The length of the curve $\xi$ can be found as 

\begin{equation}
	L = \int_{0}^{\widetilde{y}} \sqrt{ 1 + \frac{\left(\coth\left(y\right) - y\left(1+\coth^2\left(y\right)\right)\right)^2}{4y\left(\coth\left(y\right) - y\right) - 1} }\,dy = 1.50587929704238 \pm 5\times10^{-14}.
\end{equation}
Thus, the perimeter of each island is $\left(2L+\sqrt{3}\right)d\approx4.744\,d$. Notice that in this case the shape of the islands does not depend on the circulation $k$. The results are summarized in Table~\ref{tb:tab1}.

\begin{table*}[t]
    \vspace{5mm}
\caption{Summary of stagnation points and stability island characterization for the case $\Gamma_1+\Gamma_2=0$.}
\label{tb:tab1}
\begin{center}
\begin{tabular}{l|cc}
\multicolumn{3}{c}{$\Gamma_1+\Gamma_2=0$}\\
\hline

Stagnation points & $\left(\pm\frac{\sqrt3}{2},0\right)$ & unstable saddle \\
Perimeter of each island & \multicolumn{2}{c}{$\approx4.744\,d$} \\
Area of each island & \multicolumn{2}{c}{$\approx1.425\,d^2$} \\
\end{tabular}
\end{center}
\vspace*{-4pt}
\end{table*}

	\subsection{Case $\quad \Gamma_1 +\Gamma_2 \neq 0$}

We now consider the scenario where the sum of the circulations of the two vortices is not equal to zero. It is known that in such a system, the vortices will rotate around their center of vorticity, which can be defined by

\begin{equation}
C_\text{vort} = \frac{\Gamma_1z_1\left(0\right) + \Gamma_2z_2\left(0\right) }{\Gamma_1+\Gamma_2},
\end{equation}
with a velocity 
\begin{equation}
\theta = \frac{\Gamma_1+\Gamma_2}{2\pi|d|^2},
\end{equation}
where $d = z_1\left(0\right) - z_2\left(0\right)$ \cite{Batchelor, Newton, Chorin}.

Introducing the parameter $\gamma = \dfrac{\Gamma_1}{\Gamma_1+\Gamma_2}$, the equations of motion of the two vortices can be written as 


\begin{equation}
\left\{
\begin{aligned}
z_1\left(t\right) &= C_\text{vort} + \left(1-\gamma\right)d\,e^{i\theta t}\\
z_2\left(t\right) &= C_\text{vort} -\gamma\, d\,e^{i\theta t}
\end{aligned}
\right. \quad.
\end{equation}

The equation of motion for a passive particle in this system is thus
\begin{equation}
\dot{z}^*=\frac{\Gamma_1+\Gamma_2}{2\pi i}\left[\frac{\gamma}{z-z_1\left(t\right)} + \frac{1-\gamma}{z-z_2\left(t\right)}\right].
\end{equation}	

Notice that the parameter $\gamma$ defines the nature of the system. If $\gamma\in\left]0,1\right[$, both vortices will have circulations with the same sign, while if $\gamma\in\mathbb{R}\setminus\left[0,1\right]$, then the vortices have circulations of opposite signs (but not symmetric).

Consider the change of coordinates $w = \left(\dfrac{z-C_\text{vort}}{d}\right)e^{i\theta t}$. This frame is rigidly rotating with the vortices and thus the vortices are stationary and positioned at $w_1\left(t\right)= 1-\gamma$ and $w_2\left(t\right)= -\gamma$. Furthermore, since $\dot{w}^*= \dfrac{\dot{z}^*}{d^*}e^{i\theta t} + i\,\theta\, w^*$, we obtain an autonomous equation of motion for a passive particle in this frame 

\begin{equation}
\dot{w}^*=\frac{\Gamma_1+\Gamma_2}{2\pi |d|^2 i}\left[\frac{\gamma}{w-\left(1-\gamma\right)} + \frac{1-\gamma}{w+\gamma} - w^*\right].
\end{equation}.

While not as trivial as in the previous case, finding analytical expressions for the stagnation points of this system is still possible. These points should satisfy

\begin{equation}
\dot{w}^*=0 \quad\Longleftrightarrow\quad \left(1-|w|^2\right)w + \gamma\left(1-\gamma\right)w^* = \left(1-2\gamma\right)\left(1-|w|^2\right).
\end{equation}	

Using $z=x+iy$ and writing a system of equations that both the real and imaginary parts of $z$ must satisfy, we get
%

\begin{equation}
\left\{
\begin{aligned}
\left[1-\left(x^2+y^2\right)\right]x + \gamma\left(1-\gamma\right)x &= \left(1-2\gamma\right)\left[1-\left(x^2+y^2\right)\right]\\
\left[1-\left(x^2+y^2\right)\right]y - \gamma\left(1-\gamma\right)y  &= 0
\end{aligned}
\right. \quad.
\end{equation}

From the second equation, if $y\neq0$, we get $1-\left(x^2+y^2\right) = \gamma\left(1-\gamma\right)$ and, using this in the first equation, we can conclude that
\begin{equation}
y\neq 0 \Rightarrow x = \frac{1-2\gamma}{2} \Rightarrow y = \pm\frac{\sqrt{3}}{2}.
\end{equation}	
Thus, we have two stagnation points: $p_3 = \left(\dfrac{1-2\gamma}{2}, \dfrac{\sqrt{3}}{2}\right)$ and $p_4 = \left(\dfrac{1-2\gamma}{2}, -\dfrac{\sqrt{3}}{2}\right)$.

In the case where $y$ is equal to zero, we arrive at a third order equation for $x$:
\begin{equation}
x^3 + \left(2\gamma-1\right)x^2 + \left[\gamma\left(\gamma-1\right)-1\right]x + \left(1-2\gamma\right) = 0,
\label{eq:cubic}
\end{equation}	
which has either one or three real roots~\cite{discriminant}, depending on if the sign of the discriminant is negative or positive.

\begin{equation}
\Delta = -\gamma\left(\gamma-1\right)\left(3\gamma^2-3\gamma+32\right).
\end{equation}
Since $3\gamma^2-3\gamma+32>0$,$\,\,\forall\gamma$, we can not have $\Delta=0$ - which would mean that (\ref{eq:cubic}) has a repeated root - since that would only happen if $\gamma$ is $0$ or $1$. We thus conclude that 
\begin{equation}
\text{sgn}\left(\Delta\right) = \left\{
\begin{aligned}
&+1, \quad \text{if} \quad \gamma\in\left]0,1\right[\\
&-1, \quad \text{if} \quad \gamma\in\mathbb{R} \backslash \left[0,1\right]
\end{aligned}
\right. \quad.
\end{equation}
Thus, if $\gamma\in\left]0,1\right[$, there are three stagnation points along the real axis that can be written in trigonometric form as 

\begin{multline}
	x_k = -\frac{2\gamma-1}{3} + \frac{2}{3}\sqrt{\gamma^2-\gamma+4}\cos\left[\frac{1}{3}\cos^{-1}\left(\frac{\left(2\gamma-1\right)\left(\gamma^2-\gamma+16\right)}{2\left(\gamma^2-\gamma+4\right)^{3/2}}\right) - \frac{2k\pi}{3}\right], \\ k=0,1,2.
	\label{eq:x_k}
\end{multline}

If $\gamma\in\mathbb{R} \setminus \left[0,1\right]$, the sole stagnation point along the real axis can be written as
\begin{multline}
	\overline{x}_0 = -\frac{2\gamma-1}{3} +\\+\frac{2}{3}\,\text{sgn}\left(\gamma\right)\sqrt{\gamma^2-\gamma+4}\cosh\left[\frac{1}{3}\cosh^{-1}\left(\text{sgn}\left(\gamma\right)\frac{\left(2\gamma-1\right)\left(\gamma^2-\gamma+16\right)}{2\left(\gamma^2-\gamma+4\right)^{3/2}}\right)\right].
	\label{eq:x_0}
\end{multline}
As such, in addition to $p_3$ and $p_4$, if $\gamma\in\left]0,1\right[$ we have the stagnation points $p_0 = \left(x_0,0\right)$, $p_1 = \left(x_1,0\right)$ and $p_2 = \left(x_2,0\right)$, and if $\gamma\in\mathbb{R} \setminus \left[0,1\right]$, we only have $\overline{p}_0 = \left(\overline{x}_0,0\right)$ as an additional stagnation point.

 In order to get more information about the stability of each of the stagnation points and the non stationary orbits of the system, we need to analyze the equations of motion. Writing them as a system of two real ODEs, we get

\begin{equation}
\left\{
\begin{aligned}
&\dot{x} = -\frac{\gamma\,y}{\left(x-\left(1-\gamma\right)\right)^2 + y^2} - \frac{\left(1-\gamma\right) y}{\left(x+\gamma\right)^2 + y^2} + y \\
&\dot{y} =  \frac{\left(x-\left(1-\gamma\right)\right)\gamma}{\left(x-\left(1-\gamma\right)\right)^2 + y^2} + \frac{\left(1-\gamma\right) \left(x+\gamma\right) }{\left(x+\gamma\right)^2 + y^2} - x
\end{aligned}
\right. \quad,
\end{equation}
where we rescaled time as $t\to\dfrac{\Gamma_1+\Gamma_2}{2\pi |d|^2}t$ for simplicity. 

The linearization of the flow $F$ defined by the previous equations is thus

\begin{multline}
DF\left(x,y\right) = \left(\begin{matrix}
\frac{2\left(1-\gamma\right)\left(x+\gamma\right)y}{\ell_+^2} + \frac{2\gamma\left(x-\left(1-\gamma\right)\right)y}{\ell_-^2}  \\
\frac{1-\gamma}{\ell_+} + \frac{\gamma}{\ell_-} - \frac{2\left(1-\gamma\right)\left(x+\gamma\right)^2}{\ell_+^2} - \frac{2\gamma\left(x-\left(1-\gamma\right)\right)^2}{\ell_-^2} - 1 \end{matrix} \right.\\
\left.\begin{matrix}
    -\frac{1-\gamma}{\ell_+} - \frac{\gamma}{\ell_-} + \frac{2\left(1-\gamma\right)y^2}{\ell_+^2} + \frac{2\gamma y^2}{\ell_-^2} + 1\\
- \frac{2\left(1-\gamma\right)\left(x+\gamma\right)y}{\ell_+^2} - \frac{2\gamma \left(x-\left(1-\gamma\right)\right)y}{\ell_+^2}
\end{matrix}\right),
\end{multline}
where, from now on, $\ell_+\left(x,y\right) = \left(x+\gamma\right)^2 + y^2$, $\ell_-\left(x,y\right) = \left(x-\left(1-\gamma\right)\right)^2 + y^2$. 

For the points $p_3$ and $p_4$ we have:

\begin{equation}
DF\left(\frac{1-2\gamma}{2},\pm\frac{\sqrt3}{2}\right) = \begin{pmatrix}
\pm\frac{\sqrt3}{2}\left(1-2\gamma\right)

& \frac{3}{2}\\

-\frac{1}{2}

& \mp\frac{\sqrt3}{2}\left(1-2\gamma\right)
\end{pmatrix}.
\end{equation}

So the eigenvalues $\lambda_1, \lambda_2$ of $DF$ in these stagnation points must satisfy

\begin{equation}
	\left\{
	\begin{aligned}
		\lambda_1\lambda_2&=-\frac{3}{4}\left(1-2\gamma\right)^2 + \frac{3}{4}\\
		\lambda_1+\lambda_2&=0
	\end{aligned}
	\right. \quad
\end{equation}
and thus $\lambda_1=-\lambda_2=\sqrt{3\gamma\left(\gamma-1\right)}$. Therefore, if $\gamma\in\left]0,1\right[$, then $\gamma\left(\gamma-1\right)<0$ and the stagnation points are both centers; if $\gamma\in\mathbb{R} \setminus \left[0,1\right]$, then $\gamma\left(\gamma-1\right)>0$, which means that the stagnation points are both unstable saddle points.

For the points $p_0$, $p_1$ and $p_2$, which are of the form $\left(\widetilde{x},0\right)$, we have

\begin{equation}
	DF\left(\widetilde{x},0\right) = 
	\begin{pmatrix}
		0 & 1-\varphi\\
		-1-\varphi & 0
	\end{pmatrix},
\end{equation}
where $\varphi = \frac{1-\gamma}{\left(\widetilde{x}+\gamma\right)^2} + \frac{\gamma}{\left(\widetilde{x}-\left(1-\gamma\right)\right)^2}$. So, the eigenvalues $\lambda_1,\lambda_2$ of $DF\left(\widetilde{x},0\right)$ in these stagnation points must satisfy

\begin{equation}
	\left\{
	\begin{aligned}
		\lambda_1\lambda_2&=\left(1+\varphi\right)\left(1-\varphi\right)\\
		\lambda_1+\lambda_2&=0
	\end{aligned}
	\right. \quad
\end{equation}
and thus $\lambda_1=-\lambda_2=\sqrt{-\left(1+\varphi\right)\left(1-\varphi\right)}$. 

\begin{theorem}\label{T0.1} $\left(1+\varphi\right)\left(1-\varphi\right)$ is negative if $\gamma\in\left]0,1\right[$ and positive if $\gamma\in\mathbb{R} \setminus \left[0,1\right]$.
\end{theorem}

\begin{proof}
Using the the fact that $\dot{y}=0$ on the stagnation points and their analytical expressions, it is possible to infer the signs of $\left(1+\varphi\right)$ and $\left(1-\varphi\right)$ for every $\gamma$. We can write 

\begin{equation}
	\begin{aligned}
		1 + \varphi &= 1 + \frac{1-\gamma}{\left(\widetilde{x}+\gamma\right)^2} + \frac{\gamma}{\left(\widetilde{x}-\left(1-\gamma\right)\right)^2} \\
		&= 1 + \frac{1}{\left(\widetilde{x}+\gamma\right)^2} + \gamma\left[ \frac{1}{\left(\widetilde{x}-\left(1-\gamma\right)\right)^2} - \frac{1}{\left(\widetilde{x}+\gamma\right)^2} \right] \\
		&= 1 + \frac{1}{\left(\widetilde{x}+\gamma\right)^2} + \frac{\left[2\left(\widetilde{x}+\gamma\right)-1\right]\gamma}{\left(\widetilde{x}-\left(1-\gamma\right)\right)^2\left(\widetilde{x}+\gamma\right)^2}.
	\end{aligned}
	\label{eq:plus_varphi}
\end{equation}

Since $\dot{y}=0$ on the stagnation points, we have

\begin{equation}
	  \frac{\gamma}{\widetilde{x}-\left(1-\gamma\right)} + \frac{1-\gamma }{\widetilde{x}+\gamma} - \widetilde{x} = 0 \Leftrightarrow \frac{1-\gamma }{\left(\widetilde{x}+\gamma\right)^2} = 1 - \frac{\gamma}{\widetilde{x}+\gamma} - \frac{\gamma}{\left(\widetilde{x}+\gamma\right)\left(\widetilde{x}-\left(1-\gamma\right)\right)}.
\end{equation}

Using this on the expression for $1-\varphi$, we can write

\begin{equation}
	\begin{aligned}
	1 - \varphi &= 1 - \frac{1-\gamma}{\left(\widetilde{x}+\gamma\right)^2} - \frac{\gamma}{\left(\widetilde{x}-\left(1-\gamma\right)\right)^2} \\
	&= \frac{\gamma }{\widetilde{x}+\gamma} + \frac{\gamma}{\left(\widetilde{x}+\gamma\right)\left(\widetilde{x}-\left(1-\gamma\right)\right)} - \frac{\gamma}{\left(\widetilde{x}-\left(1-\gamma\right)\right)^2} \\
	&= \frac{\gamma}{\left(\widetilde{x}-\left(1-\gamma\right)\right)^2}\left[ \frac{\left(\widetilde{x}-\left(1-\gamma\right)\right)^2 }{\widetilde{x}+\gamma} + \frac{\widetilde{x}-\left(1-\gamma\right)}{\widetilde{x}+\gamma} - 1 \right] \\
	&= \frac{\gamma}{\left(\widetilde{x}-\left(1-\gamma\right)\right)^2}\left[\widetilde{x}+\gamma-2\right].
\end{aligned}
\end{equation}

If $\gamma\in\left]0,1\right[\,,$ it is easy to see that $1 + \varphi > 0\,,$ since it is the sum of three positive quantities. The sign of $1-\varphi$ will be the same as the sign of $\widetilde{x}+\gamma-2$. Using Equation (\ref{eq:x_k}), we have $\widetilde{x}+\gamma-2\leq -\frac{2\gamma-1}{3} + \frac{2}{3}\sqrt{\gamma^2-\gamma+4} + \gamma - 2 = \frac{\gamma-5}{3} + \frac{2}{3}\sqrt{\gamma^2-\gamma+4} < 0$. Thus, $\left(1+\varphi\right)\left(1-\varphi\right) < 0$. 

If $\gamma\in\mathbb{R} \setminus \left[0,1\right]$, we first need to check the signs of $\left[2\left(\widetilde{x}+\gamma\right)-1\right]\gamma$ and $\left(\widetilde{x}+\gamma-2\right)\gamma$. If $\gamma>1$, from Equation (\ref{eq:x_0}) we can write $2\left(\widetilde{x}+\gamma\right)-1 \geq -\frac{4\gamma-2}{3} +\frac{4}{3}\sqrt{\gamma^2-\gamma+4} + 2\gamma - 1 = \frac{2\gamma-1}{3} +\frac{4}{3}\sqrt{\gamma^2-\gamma+4}> 3 > 0$. Furthermore, $\widetilde{x}+\gamma-2 \geq -\frac{2\gamma-1}{3} + \frac{2}{3}\sqrt{\gamma^2-\gamma+4} + \gamma - 2 = \frac{\gamma-5}{3} + \frac{2}{3}\sqrt{\gamma^2-\gamma+4} > 0$. If $\gamma<0$, we have $2\left(\widetilde{x}+\gamma\right)-1 \leq \frac{2\gamma-1}{3} -\frac{4}{3}\sqrt{\gamma^2-\gamma+4}<-3<0$ and $\widetilde{x}+\gamma-2 \leq -\frac{2\gamma-1}{3} - \frac{2}{3}\sqrt{\gamma^2-\gamma+4} + \gamma - 2 = \frac{\gamma-5}{3} - \frac{2}{3}\sqrt{\gamma^2-\gamma+4} < -3$. Thus, we have  $\left[2\left(\widetilde{x}+\gamma\right)-1\right]\gamma>0$ and $\left(\widetilde{x}+\gamma-2\right)\gamma>0$ for both cases. Hence, $1+\varphi>0$ because every term in the right-hand side of Equation (\ref{eq:plus_varphi}) is positive and $1-\varphi>0$ because it has the same sign as $\left(\widetilde{x}+\gamma-2\right)\gamma>0$. Thus, $\left(1+\varphi\right)\left(1-\varphi\right) > 0\,.$
\end{proof}

Thus, if $\gamma\in\left]0,1\right[$, the three stagnation points in the real axis are unstable saddle points and if $\gamma\in\mathbb{R} \setminus \left[0,1\right]$, the single stagnation point in the real axis is a center. The information about the stagnation points in both scenarios is summarized in Table \ref{tb:stag points}.

\begin{table*}[t]
    \vspace{5mm}
\caption{Stagnation points for the system depending on the parameter $\gamma$.}
\label{tb:stag points}
\begin{center}
\begin{tabular}{ll|lll}
\multicolumn{2}{c}{$\gamma\in\left]0,1\right[$} & \multicolumn{2}{c}{$\gamma\in\mathbb{R} \setminus \left[0,1\right]$}\\
\hline
$\left(x_0,0\right)$, see (\ref{eq:x_k}) & unstable saddle & $\left(\overline{x}_0,0\right)$, see (\ref{eq:x_0}) & center \\

$\left(x_1,0\right)$, see (\ref{eq:x_k}) & unstable saddle & \multicolumn{2}{c}{$--$}\\

$\left(x_2,0\right)$, see (\ref{eq:x_k}) & unstable saddle & \multicolumn{2}{c}{$--$}\\

$\left(\frac{1-2\gamma}{2},\frac{\sqrt{3}}{2}\right)$ & center & $\left(\frac{1-2\gamma}{2},\frac{\sqrt{3}}{2}\right)$ & unstable saddle\\

$\left(\frac{1-2\gamma}{2},-\frac{\sqrt{3}}{2}\right)$ & center & $\left(\frac{1-2\gamma}{2},-\frac{\sqrt{3}}{2}\right)$ & unstable saddle\\

\end{tabular}
\end{center}
\vspace*{-4pt}
\end{table*}

Once again, it is possible to write a time-independent stream function for this system:
\begin{equation}
\psi\left(x,y\right) = -\frac{1}{2}\left[\gamma\log\left(\left(x-\left(1-\gamma\right)\right)^2 + y^2\right) + \left(1-\gamma\right)\log\left(\left(x+\gamma\right)^2 + y^2\right) - \left(x^2+y^2\right)\right].
\end{equation}

The trajectories of passive particles are thus the level sets of this stream function, or, equivalently, the level sets of
\begin{equation}
\widetilde{\psi}\left(x,y\right) = e^{-2\psi\left(x,y\right)} = \left[\left(x-\left(1-\gamma\right)\right)^2 + y^2\right]^{\gamma}  \left[\left(x+\gamma\right)^2 + y^2\right]^{1-\gamma} e^{-\left(x^2+y^2\right)},
\label{eq:mod_stream}
\end{equation}
which, for the symmetric case $\gamma=\frac{1}{2}$, can be thought of Cassini ovals \cite{Basset} that have been deformed by the rotational movement of the system: the curves are compressed along the $y$ axis and elongated along the $x$ axis due to the presence of the exponential term $e^{-\left(x^2+y^2\right)}$ in $\widetilde{\psi}\left(x,y\right)$.

Using the stream function and the information on the stagnation points, it is possible to draw the phase diagram of this system for any $\gamma$. We see once again that there exists a stability island around each vortex. The boundary of these islands can be characterized as a part of the curve that is implicitly defined by $\widetilde{\psi}\left(x,y\right) = \widetilde{\psi}\left(x_1,0\right)$ in the $\gamma\in\left]0,1\right[$ case, or $\widetilde{\psi}\left(x,y\right) = \widetilde{\psi}\left(p_3\right)$ in the $\gamma\in\mathbb{R} \setminus \left[0,1\right]$ case. In fact, the curves resulting from these equations separate the 2D plane in various regions. Typically, two passive particles inside of the same region will have similar trajectories, while two particles that are located in two different regions will have distinct trajectories. For instance, particles inside the stability island of one vortex will rotate around that island's vortex; particles that are far away from the vortices will have trajectories similar to the trajectory of a particle rotating around a single vortex in the vorticity center of the system, while other passive particles can have more complex trajectories where they orbit both vortices or none of them. As examples, we plot the phase diagrams for $\gamma=\frac{1}{2}$ in Figure \ref{fig:2}, $\gamma=\frac{3}{4}$ in Figure \ref{fig:3} and $\gamma=\frac{4}{3}$ in Figure \ref{fig:4}. Notice that there exists a significant distinction in the dynamics of the system depending on the value of $\gamma$ and that the case $\gamma\in\mathbb{R} \setminus \left[0,1\right]$ shows some dynamical similarities to the case $\Gamma_1 +\Gamma_2 = 0$.

Due to the complexity of the expression for the stream function and its' dependence on the parameter $\gamma$, it is not possible to obtain an expression for the area or perimeter of the stability islands of the vortices. However, using (\ref{eq:mod_stream}) it is possible to compute them numerically and the results are shown in Figure \ref{fig:5}, where we plot these quantities for each of the two stability islands as well as their sum for different values of $\gamma$.

\begin{figure}
	\centering
	\includegraphics[width=\linewidth]{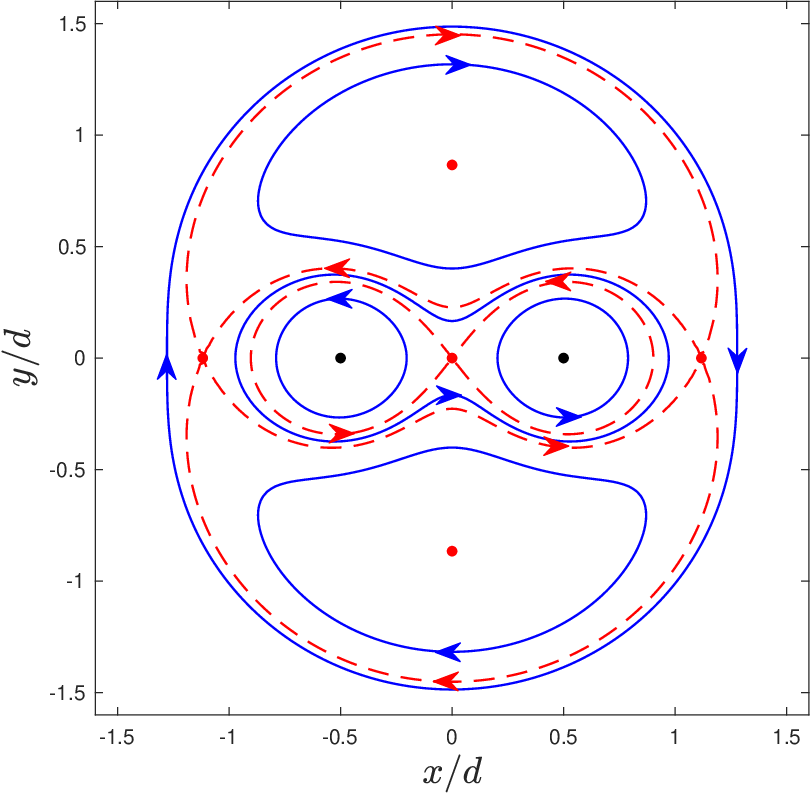}
	\caption[Phase Diagram.]{Phase diagram for a system of two point vortices in the co-moving frame. The vortices are located in $w_1=1/2$ and $w_1=-1/2$ and the circulations are such that $\Gamma_1 +\Gamma_2 \neq 0$ and $\gamma=1/2$. The red dots identify the stagnation points of the system. The red dashed lines correspond to special particle trajectories that constitute a boundary between regions of the infinite plane where particles present different types of trajectories (exemplified in blue lines).}
	\label{fig:2}
\end{figure}	

\begin{figure}
	\centering
	\includegraphics[width=\linewidth]{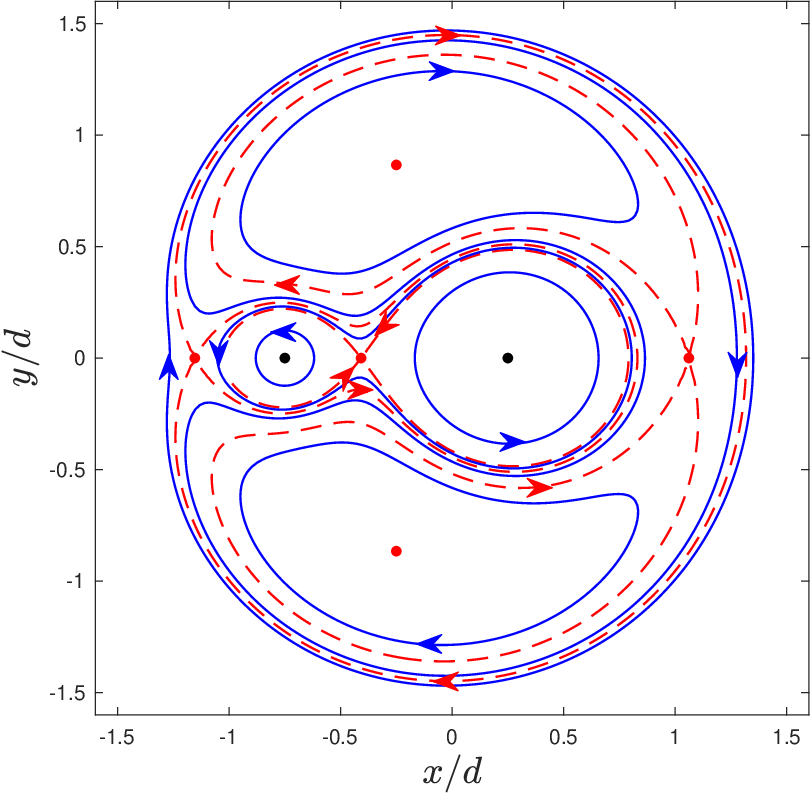}
	\caption[Phase Diagram.]{Phase diagram for a system of two point vortices in the co-moving frame. The vortices are located in $w_1=1/4$ and $w_1=-3/4$ and the circulations are such that $\Gamma_1 +\Gamma_2 \neq 0$ and $\gamma=3/4$. The red and black dashed lines correspond to special particle trajectories that constitute a boundary between regions of the infinite plane where particles present different types of trajectories (exemplified in blue lines).}
	\label{fig:3}
\end{figure}	

\begin{figure}
	\centering
	\includegraphics[width=\linewidth]{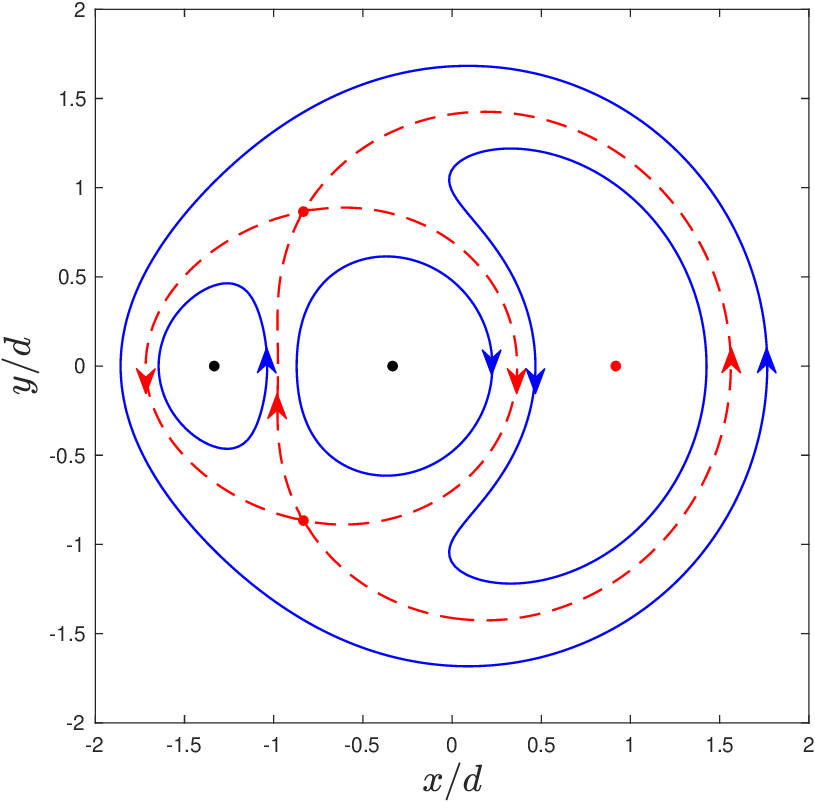}
	\caption[Phase Diagram.]{Phase Diagram for a system of two point vortices in the co-moving frame. The vortices are located in $w_1=-1/2$ and $w_1=-3/2$ and the circulations are such that $\Gamma_1 +\Gamma_2 \neq 0$ and $\gamma=4/3$. The red and black dashed lines correspond to special particle trajectories that constitute a boundary between regions of the infinite plane where particles present different types of trajectories (exemplified in blue lines).}
	\label{fig:4}
\end{figure}	

\begin{figure}
	\centering
	\includegraphics[width=0.45\linewidth]{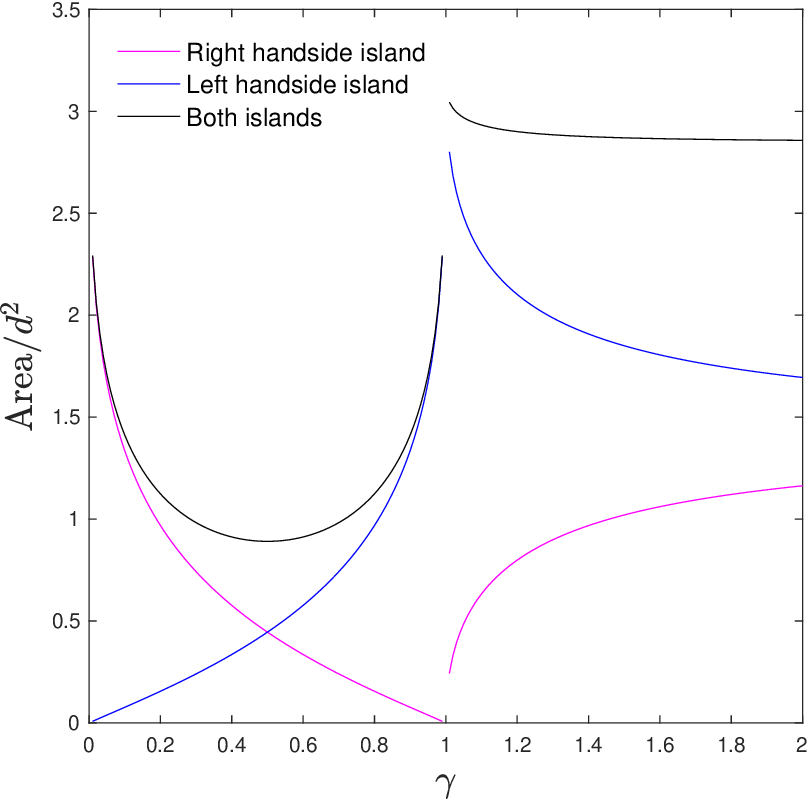}
	\includegraphics[width=0.45\linewidth]{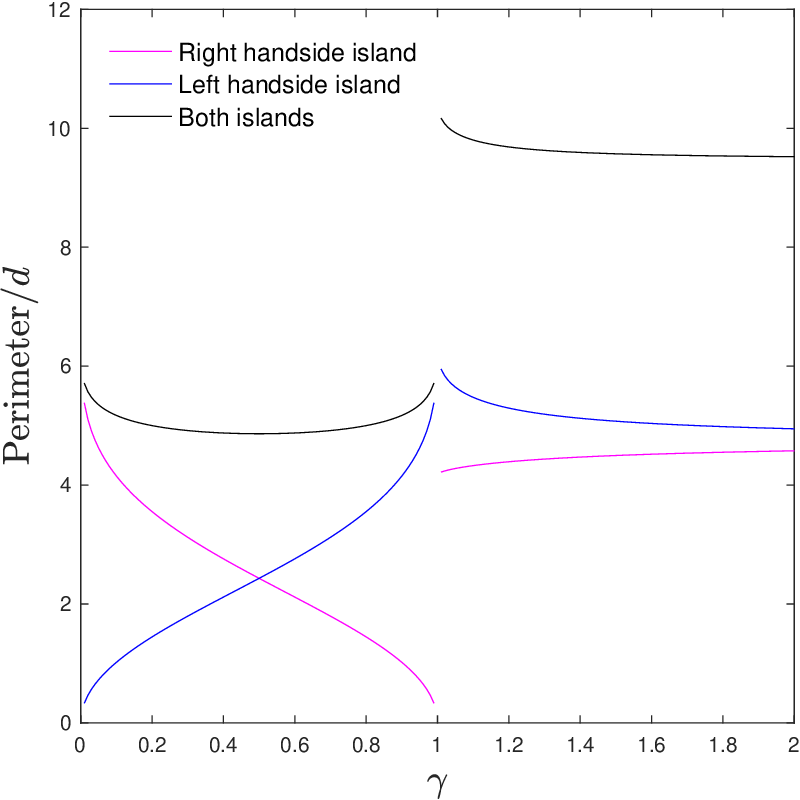}
	\caption[Phase Diagram.]{Area and Perimiter of the islands around the vortices for $0<\gamma<2$.}
	\label{fig:5}
\end{figure}	

\section{Conclusion}

Inspired by studies that have shown the existence of stability islands around point vortices in 2D flows, we study and characterize these islands for systems of two point vortices. Passive particles in a flow cannot cross the boundary of these stability islands and on the inside of the islands, particles are mainly advected by the vortex on the inside of that island, having a regular trajectory. 

In the scenario where the total circulation of the system is zero ($\Gamma_1 +\Gamma_2 = 0$), the dynamics can be completely characterized without any parameter dependence. In the frame that is co-moving with the vortices, passive particles on the outside of the stability islands will eventually move away from the vortices and never approach them again. Passive particles on the inside of he island of stability of each vortex will describe closed trajectories around them indefinitely. Due to the symmetries in the system, both these islands have the same perimeter and area, which can be calculated with arbitrary precision. 

If the total circulation is not zero, however, the dynamics are completely characterized by a parameter $\gamma=\dfrac{\Gamma_1}{\Gamma_1+\Gamma_2}$. If $\gamma\in\left]0,1\right[$, in the frame that is co-moving with the vortices, passive particles can have a multitude of trajectories: if they are on the inside of a stability island, they will describe closed trajectories around that island's associated vortex; on a certain vicinity of the vortices, particles may orbit both vortices in closed loops or describe closed trajectories without orbiting any of the vortices; and if a particle is sufficiently far away from the vortices, it will describe a closed trajectory around both of the vortices. The area and perimeter of the stability islands cannot be computed analytically due to the complexity of the mathematical expressions involved but numerical computations clearly show that these quantities scale with the strength of the vortex located on the inside of the island.

If $\gamma\in\mathbb{R} \setminus \left[0,1\right]$, in the frame that is co-moving with the vortices, particles describe closed, regular trajectories on the inside of the stability islands of the vortices once again. There is a third closed region in the 2D plane where particles describe closed trajectories without orbiting any of the vortices, while on the rest of the 2D plane, every particle will orbit both vortices in a closed trajectory. Analytical computation of the area and perimeter of the stability islands is once again impossible due to the complexity of the mathematical expressions but numerical computations show that as $\gamma\to\infty$ (equivalently $\Gamma_1\to-\Gamma_2$), the values for these quantities approach the values obtained in the $\Gamma_1 +\Gamma_2 = 0$ analysis, as one would expect.

\vskip6pt
\section*{Acknowledgements}
This work was supported by (i) CMUP, member of LASI, which is financed by national funds through FCT – Fundação para a Ciência e a Tecnologia, I.P., under the project with reference UIDB/00144/2020, and (ii) project SNAP NORTE-01-0145-FEDER-000085, co-financed by the European Regional Development Fund (ERDF) through the North Portugal Regional Operational Programme (NORTE2020) under Portugal 2020 Partnership Agreement. GM thanks grant ref. PD/BD/150537/2019 through FCT.


\vskip2pc

\bibliographystyle{RS} 
\bibliography{RSPA_Author_tex.bib} 

\begin{thebibliography}{99}

\bibitem{Helmholtz}
Helmholtz H. 1858  Über {I}ntegrale der hydrodynamischen {G}leichungen, welche
  den {W}irbelbewegungen entsprechen.. \textbf{1858}, 25--55.

\bibitem{Kelvin1869}
Thomson (Lord~Kelvin) W. 1869  On vortex motion.. {\em Trans. R. Soc. Edin}
  \textbf{25}, 217–260.

\bibitem{Kirchhoff1876}
Kirchhoff GR. 1876  Vorlesungenb\"er mathematische Physik.. {\em Mechanik}.

\bibitem{Makoto2006}
Umeki M. 2006  Clustering Analysis of Periodic Point Vortices with the $L$
  Function. {\em Journal of the Physical Society of Japan} \textbf{76}.
(\href{http://dx.doi.org/10.1143/JPSJ.76.043401}{10.1143/JPSJ.76.043401})

\bibitem{Stremler2010}
{Stremler} MA. 2010  {On relative equilibria and integrable dynamics of point
  vortices in periodic domains}. {\em Theoretical and Computational Fluid
  Dynamics} \textbf{24}, 25--37.
(\href{http://dx.doi.org/10.1007/s00162-009-0156-z}{10.1007/s00162-009-0156-z})

\bibitem{crowdy2006point}
Crowdy D. 2006  Point vortex motion on the surface of a sphere with
  impenetrable boundaries. {\em Physics of Fluids} \textbf{18}, 036602.

\bibitem{newton2010n}
Newton PK. 2010  The N-vortex problem on a sphere: geophysical mechanisms that
  break integrability. {\em Theoretical and Computational Fluid Dynamics}
  \textbf{24}, 137--149.

\bibitem{mokhov2020point}
Mokhov II, Chefranov SG, Chefranov AG. 2020  Point vortices dynamics on a
  rotating sphere and modeling of global atmospheric vortices interaction. {\em
  Physics of Fluids} \textbf{32}, 106605.

\bibitem{Nava2014}
Nava-Gaxiola C, Montaldi J. 2014  Point vortices on the hyperbolic plane. {\em
  Journal of Mathematical Physics} \textbf{55}.
(\href{http://dx.doi.org/10.1063/1.4897210}{10.1063/1.4897210})

\bibitem{Ragazzo2017}
Ragazzo C. 2017  The motion of a vortex on a closed surface of constant
  negative curvature. {\em Proceedings of the Royal Society A: Mathematical,
  Physical and Engineering Science} \textbf{473}, 20170447.
(\href{http://dx.doi.org/10.1098/rspa.2017.0447}{10.1098/rspa.2017.0447})

\bibitem{Aref_Playground}
Aref H. 2007  Point vortex dynamics: A classical mathematics playground. {\em
  Journal of Mathematical Physics} \textbf{48}, 065401.
(\href{http://dx.doi.org/10.1063/1.2425103}{10.1063/1.2425103})

\bibitem{Babiano}
Babiano A, Boffetta G, Provenzale A, Vulpiani A. 1994  Chaotic advection in
  point vortex models and two‐dimensional turbulence. {\em Phys. of Fluids}
  \textbf{6}, 2465--2474.
(\href{http://dx.doi.org/10.1063/1.868194}{10.1063/1.868194})

\bibitem{Batchelor}
Batchelor GK. 2000 {\em An Introduction to Fluid Dynamics}.
Cambridge Mathematical Library. Cambridge University Press.
(\href{http://dx.doi.org/10.1017/CBO9780511800955}{10.1017/CBO9780511800955})

\bibitem{Pope}
Pope SB. 2000 {\em Turbulent Flows}.
Cambridge University Press.
(\href{http://dx.doi.org/10.1017/CBO9780511840531}{10.1017/CBO9780511840531})

\bibitem{boatto_pierrehumbert_1999}
BOATTO S, PIERREHUMBERT RT. 1999  Dynamics of a passive tracer in a velocity
  field of four identical point vortices. {\em Journal of Fluid Mechanics}
  \textbf{394}, 137–174.
(\href{http://dx.doi.org/10.1017/S0022112099005492}{10.1017/S0022112099005492})

\bibitem{Newton}
Newton P. 2001 {\em The N-Vortex Problem: Analytical Techniques}.
Appl. Math. Sci. Springer New York.

\bibitem{Chorin}
Chorin A. 1994 {\em Vorticity and Turbulence}.
Appl. Math. Sci. Springer.

\bibitem{discriminant}
Nickalls RWD, Dye RH. 1996  The Geometry of the Discriminant of a Polynomial.
  {\em The Mathematical Gazette} \textbf{80}, 279--285.

\bibitem{Basset}
Basset AB. 1901 {\em An elementary treatise on cubic and quartic curves}.
University of Michigan Library.

\end{thebibliography}

\end{document}